\documentclass[10pt, conference]{IEEEtran}
\IEEEoverridecommandlockouts
\usepackage{pifont}
\usepackage{tabu}
\usepackage{supertabular,booktabs}

\usepackage{longtable}

\usepackage[hyphens]{url}
\usepackage{multicol}
\usepackage{amsmath, amssymb, bm, cite, epsfig, psfrag, mathtools}
\usepackage{epstopdf}
\usepackage{changepage}
\usepackage{graphicx}
\usepackage{fancyhdr}
\usepackage{bbm}
\newcommand*{\Scale}[2][4]{\scalebox{#1}{$#2$}}%
\usepackage{algorithm,algorithmic}
\usepackage{array}
\usepackage[margin=10pt,font=small]{caption}
\usepackage{bbm}
\usepackage{multirow}
\usepackage[usenames,dvipsnames]{xcolor}

\usepackage{subfigure}
\usepackage{etoolbox}
\usepackage{pbox}
\usepackage[width=0.48\textwidth]{caption}
\usepackage{colortbl}
\usepackage{bm}
\usepackage{caption}
\usepackage{amsmath}
\usepackage{color}

\usepackage{enumitem}
\usepackage{datetime}
\usepackage{amsthm}
\usepackage{dsfont}
\usepackage{mathtools}
\usepackage[T1]{fontenc}

\newtoggle{conference}
\togglefalse{conference} 
\usepackage[letterpaper,left=0.75in,right=0.75in,top=0.75in,bottom=0.75in,footskip=.25in]{geometry}
\graphicspath{{figures/}}

\def\beq{\begin{equation}}
\def\eeq{\end{equation}}
\def\beqa{\begin{eqnarray}}
\def\eeqa{\end{eqnarray}}
\def\beqan{\begin{eqnarray*}}
\def\eeqan{\end{eqnarray*}}

\def\EE{{\mathbb{E}}}
\def\PP{{\mathbb{P}}}

\def\argmin{\mathop{\mathrm{arg\,min}}}
\def\argmax{\mathop{\mathrm{arg\,max}}}

\newcommand{\Ic}{{\cal I}}
\newcommand{\Jc}{{\cal J}}

\newtheorem{Remark}{Remark}

\newtheorem{definition}{Definition}
\newtheorem{theorem}{Theorem}

\newtheorem{example}{Example}

\def\FF{{\mathbb{F}}}

\def\tm1{t\! - \! 1}
\def\tp1{t\! + \! 1}

\def\fbf{\mathbf{f}}

\def\pbf{\mathbf{p}}

\def\qbf{\mathbf{q}}

\def\Cbf{\mathbf{C}}

\newcommand{\Ec}{{\cal E}}
\def\Gbf{\mathbf{G}}

\def\Qbf{\mathbf{Q}}

\def\Uc{\mathcal{U}}

\newcommand{\Vc}{{\cal V}}
\newcommand{\Kc}{{\cal K}}

\begin{document}

\bibliographystyle{IEEEtran}


\title{ \fontsize{23.65}{16}\selectfont Cache-Aided Coded Multicast for Correlated Sources}

\author{P. Hassanzadeh, A. Tulino, J. Llorca, E. Erkip
\thanks{P. Hassanzadeh  and  E. Erkip are with the ECE Department of New York University, Brooklyn, NY. Email: \{ph990, elza\}@nyu.edu}
\thanks{J. Llorca  and A. Tulino are with Bell Labs, Nokia, Holmdel, NJ, USA. Email:  \{jaime.llorca, a.tulino\}@nokia.com}
\thanks{A. Tulino is with the DIETI, University of Naples Federico II, Italy. Email:  \{antoniamaria.tulino\}@unina.it}
}

\maketitle
\begin{abstract}
The combination of 
edge caching and coded multicasting is a promising approach to improve the efficiency of content delivery over cache-aided 
networks.
The global caching gain resulting from content overlap distributed across the network in current 
solutions is limited due to the increasingly personalized nature of the content consumed by users. 
In this paper, the cache-aided coded multicast problem is generalized to account for the correlation among the network content by formulating a source compression problem with distributed side information.
A correlation-aware achievable scheme is proposed and an upper bound on its performance is derived.
It is shown that considerable 
load reductions can be achieved, compared to state of the art correlation-unaware schemes, when caching and delivery phases specifically account for the correlation among the content files.


\end{abstract}

\section{Introduction}~\label{sec:Introduction}
Proper distribution of popular content across the network caches is emerging as one of the promising approaches to address the exponentially growing traffic in current wireless networks. Recent studies \cite{maddah14fundamental,maddah14decentralized,ji14average,ji15order,ji15groupcast,ji15efficient} have shown that, in a cache-aided network, exploiting globally cached information in order to multicast coded messages that are useful to a large number of receivers exhibits overall network throughput that is proportional to the aggregate cache capacity. The fundamental rate-memory trade-off in a broadcast caching network has been characterized in \cite{maddah14fundamental,maddah14decentralized,ji14average,ji15order,ji15groupcast,ji15efficient }. While these initial results are promising, these studies treat the network content as independent pieces of information, and do not account for the additional potential gains arising from further compression of correlated content distributed across the network.
,
In this paper, we investigate how the correlations among the library content can be explored in order to further improve the performance in cache-aided networks. ,We consider a network setup similar to \cite{maddah14fundamental,maddah14decentralized,ji14average,ji15order,ji15groupcast,ji15efficient}, but assume that the files in the library are correlated. Such correlations are especially relevant among content files of the same category, such as episodes of a TV show or same-sport recordings, which, even if personalized, may share common backgrounds and scene objects.

As in existing literature on cache-aided networks, 
we assume that the network operates in two phases: a caching (or placement) phase taking place at network setup followed by a delivery phase where the network is used repeatedly in order to satisfy receiver demands. The design of the caching and delivery phases forms what is referred to as a {\em caching scheme}. 
During the caching phase, caches are filled with content from the library according to a properly designed {\em caching distribution}. During the delivery phase, the sender compresses the set of requested files into a multicast codeword by computing an {\em index code} while exploring all correlations that exist among the requested and cached content.
In \cite{timo2015rate}, the rate-memory region for a correlated library was characterized for lossy reconstruction in a scenario with only two receivers and a single cache. In this paper, we consider correlation-aware lossless reconstruction in a more general setting with multiple caches, and with a nonuniform demand distribution.

We propose an achievable correlation-aware scheme, named Correlation-Aware RAP caching and Coded Multicast delivery (CA-RAP/CM), in which receivers store content pieces based on their popularity as well as on their correlation with the rest of the file library during the caching phase, and receive compressed versions of the requested files according to the information distributed across the network and their joint statistics during the delivery phase.
The scheme consists of: 1) exploiting file correlations to store the {\em more relevant bits} during the caching phase such that expected delivery rate is reduced, and 2) optimally designing the coded multicast codeword based on the joint statistics of the library files and the aggregate cache content during the delivery phase. 
Additional refinements are transmitted, when needed, in order to ensure lossless reconstruction of the requested files at each receiver.
Given the  exponential complexity of CA-RAP/CM, we then provide an algorithm which approximates  CA-RAP/CM in polynomial time, and we derive an upper bound on the achievable expected rate. We numerically compare the rates achieved by the proposed correlation-aware scheme with existing correlation-unaware schemes, and our results confirm the additional gains achievable, especially for small memory size.



The paper is organized as follows. The problem formulation is presented in Sec. \ref{sec:Problem Formulation}. In  Sec. \ref{sec:Achievable Scheme} we describe CA-RAP/CM  and  its polynomial-time approximation, and we quantify the associated rate-memory trade-off.
We provide numerical simulations and concluding remarks in Secs. \ref{sec:Simulations} and \ref{sec:Conclusion}, respectively.

\section{Network Model and Problem Formulation}\label{sec:Problem Formulation}

{\bf Notation:} For ease of exposition, we use $\{A_i\}$ to denote the set of elements $\{A_i:i\in\mathcal I\}$, with $\mathcal I$ being the domain of index $i$, and we define $[n]\triangleq\{1,\dots,n\}$. $\FF_2^*$ denotes the set of finite length binary sequences.

Similar to previous work \cite{ji14average,ji15order}, we consider a broadcast caching network composed of one sender (e.g., base station) with access to a file library composed of $m$ files,  
each of entropy $F$ bits. 
Without loss of generality, we assume that each file $f\in [m]$ is represented by a vector of i.i.d. random binary symbols of length $F$, ${\sf W}_f \in \FF_2^F$, { and hence $H({\sf W}_f) = F$ for all $f\in [m]$.}\footnote{ The results derived in this paper can be easily extended to the case in which the symbols within each file are not necessarily binary and independent.}
We assume that the symbols belonging to different files can be correlated, i.e.,  $H({\sf W}_1,\dots,{\sf W}_m)\leq mF$, and we denote their joint distribution by $P_{\mathcal W}$.
The sender communicates with $n$ receivers (e.g., access points or user devices) $\mathcal U=\{1,\dots,n\}$ through a shared error-free multicast link. Each receiver has a cache of size $MF$ bits and requests files in an independent and identically distributed (i.i.d.) manner according to a demand distribution $\qbf = (q_1,\dots,q_m)$, where $q_f$ denotes the probability of requesting file  $f\in [m]$. 
A caching scheme for this network consists of:

\begin{itemize}
\item {\textbf{Cache Encoder:}}
Given a library realization $\{W_f\}$, the cache encoder at the sender computes the cache content $Z_u(\{W_f\})$ at receiver $u \in \mathcal U$, using the set of functions
$\{Z_u: \FF_2^{m F} \rightarrow \FF_2^{MF} : u \in \Uc\}$.
The encoder designs the cache configuration $\{Z_u\}$ jointly across receivers, taking into account global system knowledge such as the number of receivers and their cache sizes, the number of files, their aggregate popularity, and their joint distribution $P_{\mathcal W}$.

\item{\textbf{Multicast Encoder:}}
After the caches are populated, the network is repeatedly used for different demand realizations. At each use of the network, the receivers place a random demand vector 
$\bold{f} \in[m]^{n}$
at the sender, 
according to $\qbf$. The demand realization is denoted by $\fbf =(f_1,\dots,f_n)$.
The multicast encoder at the sender computes codeword $X(\fbf, \{W_f\},\{Z_u\})$, which will be transmitted over the shared link, 
using the function $X : [m]^{n} \times  \FF_2^{mF} \times \FF_2^{nMF} \rightarrow \FF_2^*$. In this work, we consider a fixed-to-variable almost-lossless framework.

\item{\textbf{Multicast Decoders:}}
Receivers recover their requested files using their cached content and the received multicast codeword. More specifically, receiver $u$ recovers $W_{f_u}$ using its decoding function $\zeta_u : [m]^n \times \FF_2^* \times \FF_2^{MF}  \rightarrow \FF_2^{F}$,
as $\widehat{W}_{f_u} = \zeta_u(\fbf, X,Z_u)$.

\end{itemize}

The worst-case (over the file library) probability of error of the corresponding caching scheme is defined as
\begin{align} \label{perr}
& P_e^{(F)} = \sup_{\{W_f : f \in [m]\}} \PP \left(\widehat{W}_{f_u}  \neq W_{f_u} \right). \notag
\end{align}
In line with previous work \cite{ji14average,ji15order},
the (average) rate of the overall {caching scheme} is defined as
\begin{equation} \label{average-rate}
R^{(F)} = \sup_{\{ W_f : f \in [m] \}} \; \frac{\EE[J(X)]}{F}.
\end{equation}
where $J(X)$ denotes the length (in bits) of the multicast codeword $X$.
\begin{definition} \label{def:achievable-rate}
A rate $R$ is {\em achievable} if there exists a sequence of caching schemes for increasing file size $F$ such that
{$\lim_{F \rightarrow \infty} P_e^{(F)} = 0 \notag$, and $\limsup_{F \rightarrow \infty} R^{(F)} \leq  R.$}
\end{definition}

The goal of this paper is to design a caching scheme that results in a small achievable rate $R$.

\section{Correlation-Aware RAP Caching and Coded Multicast Delivery (CA-RAP/CM)}\label{sec:Achievable Scheme}
In this section, we introduce CA-RAP/CM, a correlation-aware caching scheme, which is an extension of a RAndom Popularity-based (RAP) caching policy followed by a Chromatic-number Index Coding (CIC) delivery policy \cite{ji14average, ji15order}. In CA-RAP/CM, both the cache encoder and the multicast encoder are designed according to the joint distribution $P_{\mathcal W}$, in order to exploit the correlation  among the library files. First, consider the following motivating example. 
\begin{example}\label{ex:CA RAP}
Consider a file library with $m=4$ uniformly popular files $\{W_1,W_2,W_3,W_4\}$ each with entropy $F$ bits. 
We assume that the pairs $\{W_1 ,W_2\}$ and $\{W_3, W_4\}$ are independent, while correlations exist between $W_1$ and $W_2$, and between $W_3$ and $W_4$.
Specifically, $H(\sf W_1|\sf W_2)=H(\sf W_2|\sf W_1)=F/4$ and $H(\sf W_3|\sf W_4)=H(\sf W_4|\sf W_3)=F/4$.
The sender is connected to $n=2$ receivers $\{u_1, u_2\}$ with cache size $M=1$.
While a correlation-unaware scheme (e.g., \cite{maddah14decentralized,ji14average}) would first compress the files separately and then cache $1/4^{th}$ of each file at each receiver, existing file correlations can be exploited to cache the more relevant bits. 
For example, one could split
files $W_2$ and $W_4$ into two parts $\{ W_{2,1}, W_{2,2}\}$ and $\{ W_{4,1}, W_{4,2}\}$, each with entropy $F/2$, and cache 
  $\{ W_{2,1}, W_{4,1}\}$ at $u_1$ and 
  $\{ W_{2,2}, W_{4,2}\}$ at $u_2$, as shown in Fig. \ref{fig:Examples}. 
During the delivery phase, consider the worst case demand,
e.g., $\fbf = (W_3,W_1)$, 
the sender first multicasts the XOR of the compressed parts $W_{2,1}$ and $W_{4,2}$. Refinement segments, with refinement rates $H(\sf W_3|\sf W_4)$ and $H(\sf W_1|\sf W_2)$ can then be transmitted to enable lossless reconstruction, resulting in 
a total rate $R=1$.
Note that a correlation-unaware scheme 
would need a total rate $R=1.25$ regardless of the demand realization \cite{maddah14decentralized,ji14average}.
\end{example}
%

\begin{figure}
  \centering
  \includegraphics[width=2.7in]{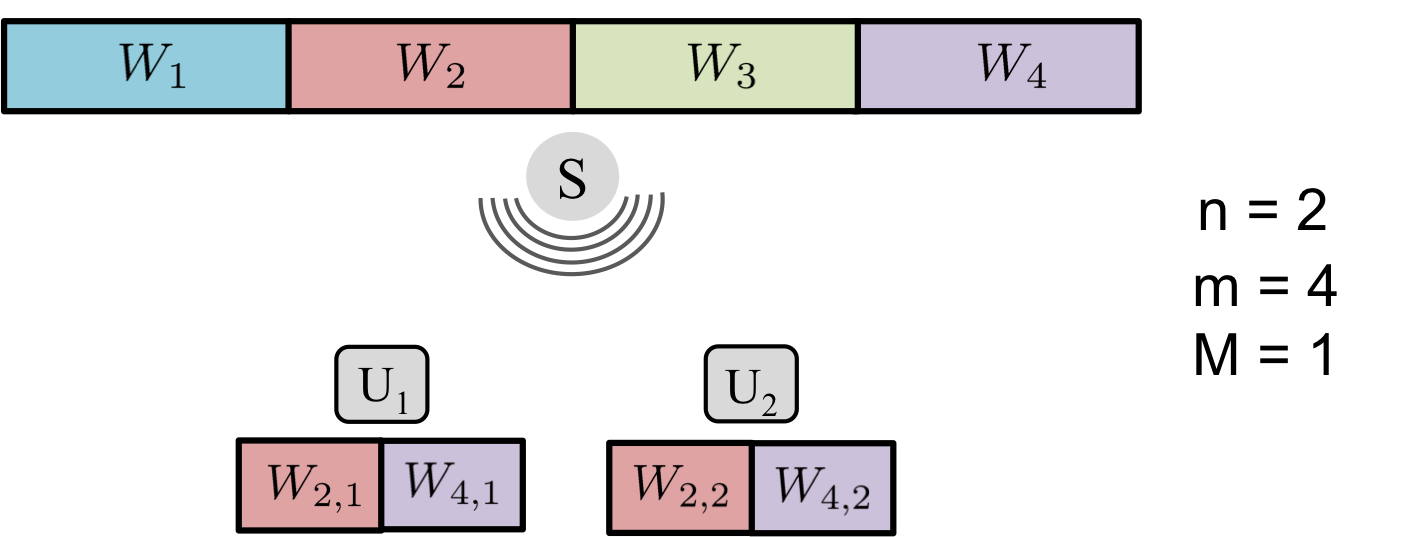}
  \caption{Network setup and cache configuration of Example \ref{ex:CA RAP}.} 
  \label{fig:Examples}
\end{figure}

The main components of the CA-RAP/CM scheme are: $i$) a {\em Correlation-Aware Random Popularity Cache Encoder} (CA-RAP) and  $ii$) a {\em Correlation-Aware Coded Multicast Encoder} (CA-CM).

\subsection{Correlation-Aware Random Popularity Cache Encoder}\label{subsec:CA RAP}

The CA-RAP cache encoder is a correlation-aware random fractional cache encoder whose key differentiation from the cache encoder RAP, introduced in \cite{ji14average, ji15order}, is to choose the fractions of files to be cached according to both their popularity as well as their correlation with the rest of the library. 
Similar to the cache encoder RAP, each file is partitioned into $B$ equal-size packets, with packet $b\in[B]$ 
of file $f\in[m]$ denoted by $W_{f,b}$. 
The cache content at each receiver is selected according to a \textit{caching distribution}, $\pbf = (p_{1}, \dots, p_{m})$ with $0 \leq p_{f} \leq 1/M$ $\forall f \in[m]$ and $\sum_{f = 1}^{m} p_{f} = 1$, which is optimized 
to minimize the rate of the corresponding index coding delivery scheme.  
For a given caching distribution $\pbf$, each receiver caches a subset of $p_{f}MB$ distinct packets from each file $f \in [m]$, independently at random. We denote by $\Cbf=\{ \Cbf_1,\dots,\Cbf_n \}$ the packet-level cache configuration, where $\Cbf_{u}$ denotes the set of file-packet index pairs, $(f,b)$, $f\in[m]$, $b\in[B]$, cached at receiver $u$.
In Example \ref{ex:CA RAP}, $B=2$,  the caching distribution corresponds to  $p_{W_2}=p_{W_4}=1/2$, $p_{W_1}=p_{W_3}=0$, and the packet-level cache configuration is $\Cbf=\{\{(2,1),(4,1)\},\{(2,2),(4,2)\}\}$. 
While  the caching distribution of a correlation-unaware scheme prioritizes the caching of packets according to the aggregate popularity distribution (see \cite{ji14average, ji15order}), the CA-RAP caching distribution accounts for both the aggregate popularity and the correlation of each file with the rest of the library when determining the amount of packets to be cached from each file.

\subsection{Correlation-Aware Coded Multicast Encoder}

{For a given demand realization $\fbf$, the packet-level demand realization is denoted by $\Qbf=[\Qbf_1,\dots,\Qbf_n]$, where $\Qbf_{u}$ denotes the file-packet index pairs $(f,b)$ associated with the packets of file $W_{f_u}$ requested, but not cached, by receiver $u$. }

The CA-CM encoder capitalizes on the additional coded multicast opportunities that arise from
incorporating cached packets that are, not only equal to, but also correlated with the requested packets into the multicast codeword.
The CA-CM encoder operates by constructing a {\em clustered conflict graph}, and computing a linear index code from a valid\footnote{A valid coloring of a graph is an assignment of colors to the vertices of the graph such that no two adjacent vertices are assigned the same color.} coloring of the conflict graph, as described in the following.


\subsubsection{Correlation-Aware Clustering} \label{subsubsec:correlation-Aware Source Clustering}

For each requested packet $W_{f, b}$,  $(f,b)\in\Qbf$, the correlation-aware clustering procedure computes a {\em $\delta$-ensemble} $\mathcal G_{f,b}$
, where $\mathcal G_{f,b}$ is the union of $W_{f,b}$ and the subset of all cached and requested packets that are $\delta$-correlated with $W_{f,b}$, as per the following definition.\footnote{In practice, the designer shall determine the level at which to compute and exploit correlations between files based on performance-complexity tradeoffs. For example, in video applications, a packet may represent a video block or frame.}




\begin{definition}({\bf $\delta$-Correlated Packets})\label{def:delta-cor}
For a given threshold $\delta\leq1$, packet ${\sf W}_{f,b}$  is $\delta$-correlated with packet ${\sf W}_{f',b'}$ if $H({\sf W}_{f,b},{\sf W}_{f',b'})\leq (1+\delta)F $ bits, for all $f,f' \in [m]$ and $b,b' \in [B]$.
\end{definition}

\noindent{This classification { (clustering)} is the first step for constructing the clustered conflict graph.}


\subsubsection{Correlation-Aware Cluster Coloring} \label{subsubsec:Correlation-Aware Chromatic Cluster Covering}
The clustered conflict graph $\mathcal H_{\Cbf,\Qbf}=(\Vc, \Ec)$ is constructed as follows:

\begin{itemize}
\item
Vertex set $\Vc$: The vertex (node) set $ \Vc=\widehat{\mathcal V}\cup\widetilde{\mathcal  V}$ is composed of root nodes $\widehat{\mathcal V}$ and virtual nodes $\widetilde{\mathcal V}$.

\begin{itemize}
 \item Root Nodes: 
 There is a root node $ \hat v\in\widehat{\mathcal V}$ for each packet 
 $W_{f,b}$, 
 requested by each receiver, 
 uniquely identified by the pair  $\{ \rho(\hat v),\mu(\hat v)\}$,  with $\rho(\hat v)$ denoting the packet identity, i.e., file-packet index pair $(f,b)$, 
 and  $\mu(\hat v)$ denoting the receiver requesting it.

\item Virtual Nodes:
For each root node $\hat v \in \widehat{\mathcal \Vc}$,  all the packets in the 
$\delta$-ensemble $\mathcal G_{\rho(\hat v)}$ other than $\rho(\hat v)$ are represented as virtual nodes in $\widetilde{\mathcal V}$.
We identify virtual node $\tilde v\in\widetilde{\mathcal V}$, having $\hat v$ as a root note, with the triplet $\{ \rho(\tilde v), \mu(\hat v),r(\tilde v)\}$, where $\rho(\tilde v)$ indicates the \mbox{packet identity} associated with vertex $\tilde v$, $\mu(\hat v)$ indicated the receiver requesting $\rho(\hat v)$, and  $r(\tilde v) = \hat v$ is the root of the $\delta$-ensemble that $\tilde v$ belongs to.
\end{itemize}

We denote by $\Kc_{\hat v} \subseteq \Vc$ the set of vertices containing root node $\hat v \in \widehat{\Vc}$ and the virtual nodes corresponding to the packets in its $\delta$-ensemble  $\mathcal G_{\rho(\hat v)}$, and we refer to $\Kc_{\hat v}$ as the {\em cluster} of root node $\hat v$.

\item Edge set $\Ec$: For any pair of vertices $v_1, v_2 \in \Vc $, there is an edge between $v_1$ and $v_2$ in $\mathcal E$  if  both $v_1$ and  $v_2$ are in the same cluster, 
or if the two following conditions are jointly satisfied  1) $\rho(v_1) \neq \rho(v_2)$, 2) packet $\rho(v_1) \notin \Cbf_{\mu(v_2)}$ or packet $\rho(v_2) \notin \Cbf_{\mu(v_1)}$.

 Given a valid coloring of $\mathcal H_{\Cbf,\Qbf}$, a valid \textit{cluster coloring} of $\mathcal H_{\Cbf,\Qbf}$  consists  of assigning to each cluster $\Kc_{\hat v},\; \forall \hat v\in \widehat{\Vc}$, one of the colors assigned to the vertices inside that cluster. For each color in the cluster coloring, only the packets with the same color as the color assigned to their corresponding cluster, are XORed together and multicasted to the users. Using its cached information and the received XORed packets, each receiver is able to reconstruct a (possibly) distorted version of its requested packet, due to the potential reception of a packet that is $\delta$-correlated with its requested one. The encoder transmits refinement segments, when needed, to enable lossless reconstruction of the demand at each receiver. 
The coded multicast codeword results from concatenating: 1) for each color in the cluster coloring, the XOR of the packets with the same color, and, 2) for each receiver, if needed, the refinement segment.
The CA-CM encoder selects the valid cluster coloring corresponding to the shortest coded multicast codeword.


\end{itemize}

{\Remark Note that if correlation is not considered or non-existent, the clustered conflict graph is equivalent to the conventional index coding conflict graph \cite{ji14average, ji15order}, i.e., the subgraph of $\mathcal H_{\Cbf,\Qbf}$ resulting from considering only the root nodes $\widehat{\mathcal V}$.}

{\Remark It is important to note that the number of colors in the cluster coloring chosen by CA-CM is always smaller than or equal to the chromatic number\footnote{The chromatic number of a graph is the minimum number of colors over all valid colorings of the graph.} of the conventional index coding conflict graph. 
{ Such reduction in the number of colors is obtained by considering correlated packets that are cached in the network, which possibly results in less conflict edges and provides more options for coloring each cluster. Intuitively, CA-CM allows for the requested packets that had to be transmitted by themselves otherwise, to be represented by correlated packets that can be XORED together with other packets in the multicast codeword.}
} 

\subsection{Greedy Cluster Coloring (GClC)}\label{subsec: algorithms}
Given that graph coloring, and by extension cluster coloring, is 
NP-Hard, 
in this section we propose {\em Greedy Cluster Coloring (GClC)}, 
a polynomial-time 
approximation to the cluster coloring problem. GClC extends the existing Greedy Constraint Coloring (GCC) scheme \cite{ji15order}, to account for file correlation in cache-aided networks and  
consists of a combination of two coloring schemes, such that the scheme resulting in the lower number of colors (i.e., shorter multicast codeword) is chosen. 
Uncoded refinement segments are transmitted to ensure lossless reconstruction of the demand.


\begin{algorithm}[ht]
\caption{ GClC$_1$ - Greedy Cluster Coloring 1}
\label{algorithm: GClC 1}
\small{
\begin{algorithmic}[1]
\WHILE {$\widehat\Vc \neq \emptyset$}
    \STATE Pick any root node $\hat v \in \widehat\Vc$

    \STATE Sort $\Kc_{\hat v}$ in decreasing order of receiver label size\footnotemark, such that for $v_t,v_{t+1} \in  \Kc_{\hat v}$, $|\{\mu(v_t),\eta(v_t)\}| \geq |\{\mu(v_{t+1}),\eta(v_{t+1})\}|$ where $v_t$ denotes the $t^{th}$ vertex in the ordered sequence.
    \STATE $t = 1$
    \WHILE {$t \leq | \Kc_{\hat v}|$}
         \STATE Take $ v_t \in \Kc_{\hat v}$; Let $I_{v_t} = \{ v_t\}$
                 \FORALL{$ v \in \Vc\backslash \Big\{ \Kc_{\hat v} \cup I_{v_t}  \Big\}$}
            \IF{ $\Big\{$ There is no edge between $v$ and $I_{v_t}$ $\Big\}$ $\bigwedge$ \\
             $\Big\{ \{\mu(v),\eta(v)\} = \{\mu(v_t),\eta(v_t)\} \Big\} $ }
                \STATE $I_{v_t} = I_{v_t} \cup \{v\}$
            \ENDIF
        \ENDFOR
    \STATE ${v_t}^{*} = \argmax\limits_{\{ v_\tau:\, \tau=1, \ldots,t \}} |I_{v_\tau}|$
    \IF{ $|I_{v_t^{*}}| \geq |\{\mu(v_{t}),\eta(v_{t})\}|$  or $t=| \Kc_{\hat v}|$ }
        \STATE $\Ic = I_{v_t^{*}} $
        \STATE $t=| \Kc_{\hat v}|+1$ 
    \ELSE
        \STATE $t = t+1$
    \ENDIF
    \ENDWHILE
    \STATE Color all vertices in $\Ic$ with an unused color.
    \STATE {$\Jc = \{ r(v) : v \in \Ic\}$}
     \STATE { $\widehat \Jc = \Big \{\hat  v \!\in \!\widehat\Vc :  \exists  v \! \in \!\Ic, \, \! \Big\{ \mu(\hat v)=\mu(v) \Big\} \! \bigwedge  \!  \Big\{\rho(v) \in \mathcal G_{\rho(\hat v)}  \Big\}     \Big \}$}

    \STATE $\widehat\Vc  \leftarrow \widehat\Vc\backslash  \Jc \cup \widehat \Jc$, $\Vc \leftarrow \Vc\backslash \bigcup\limits_{\hat v \in \Jc\cup \widehat\Jc} \Kc_{\hat v}$
\ENDWHILE

\end{algorithmic}
}
\end{algorithm}
\footnotetext{Packets with equal label size are ordered such that $H({\sf W}_{\rho(\hat v)}|{\sf W}_{\rho(v_t)})\leq H({\sf W}_{\rho(\hat v)}|{\sf W}_{\rho(v_{t+1})})$ }
{ In GClC, we assume that any vertex (root node or virtual node) $v\in\mathcal V$ is identified by the triplet $\{\rho(v),\mu(v), r(v)\}$, which is uniquely specified by the \mbox{packet identity} associated with $v$ and by the cluster to which $v$ belongs.} Specifically, given a vertex $v \in   \Kc_{\hat v}$, then $\rho(v)$ indicates the \mbox{packet identity} associated with vertex $v$,  while $\mu(v)=\mu(\hat v)$  and $r(v)=\hat v$. 
Further define $\eta(v)\triangleq\{u\in\mathcal U: \rho(v) \in \Cbf_u\}$ for any $v\in \Vc$. 
The unordered set of receivers $\{\mu(v), \eta(v)\}$, corresponding to the set of receivers either requesting or caching packet $\rho(v)$, is referred to as the {\em receiver label} of vertex $v$.

Algorithm GClC$_1$ starts from a root node $\hat v\in\widehat\Vc$ among those not yet selected, and searches for the node $v_t \in \Kc_{\hat v}$ 
which forms the largest independent set $\Ic$  with all the vertices in $ \Vc$ having its same receiver label.\footnote{An independent set is a set of vertices in a graph, no two of which are adjacent.} Next, vertices in set $\Ic$ are assigned the same color (see lines 20-23). 
Algorithm GClC$_2$ is based on a correlation-aware extension of GCC$_2$ in \cite{ji15order}, 
and corresponds to a generalized uncoded (naive) multicast: 
For each root node  $\hat v \in \widehat\Vc$, whose cluster has not yet been colored, 
only the vertex  $ v_t\in \mathcal K_{\hat v}$ whom is found among the nodes of more clusters, i.e., correlated with a larger number of requested packets, is colored and its color is assigned to $\mathcal K_{\hat v}$ and all clusters containing $v_t$.
For both GClC$_1$ and GClC$_2$,  when the graph coloring algorithm terminates, only a subset of the graph vertices, $\mathcal V$, are colored such that only one vertex from each cluster in the graph is colored. This is equivalent to identifying a valid cluster coloring where each cluster is assigned the color of its colored vertex.
Between GClC$_1$ and GClC$_2$, the cluster coloring resulting in the lower number of colors is chosen. For each color assigned during the coloring, the packets with the same color are XORed together, and multicasted.

\subsection{Performance of CA-RAP/CM}~\label{subsec:correlation aware rate}

In this section, we provide an upper bound of the rate achieved with CA-RAP/CM. 
For a given $\delta$, we define  the {\em match matrix}  $\Gbf$ as the  matrix whose element $\Gbf_{f'f}$ $(f, f') \in [m]^2$ is the largest value such that for each packet $W_{f,b}$ of file $f$, there are at least $\Gbf_{f'f}$ packets of file $f'$ that are $\delta$-correlated with $W_{f,b}$, and are distinct from the packets correlated with packet $W_{f,b'}$, $\forall \, b'\in[B]$.

\begin{theorem}\label{thm:general CA-RAP/CM}
Consider a broadcast caching network with $n$ receivers, cache capacity $M$, demand distribution $\qbf$, a caching distribution $\pbf$, library size $m$, correlation parameter $\delta$,  and match matrix $\Gbf$. Then, the achievable expected rate of CA-RAP/CM, $R(\delta , \pbf)$, is upper bounded, as $F \rightarrow \infty$, with high probability as
\begin{align}
R(\delta , \pbf)  \leq  \min \left \{\psi(\delta , \pbf) + \Delta R(\delta , \pbf),\bar m \right \}, 
\end{align}
where
\begin{align}
& \psi(\delta , \pbf) \triangleq  \sum\limits_{\ell =1 }^{n} \binom{n}{\ell} \sum\limits_{f =1}^m \rho_{\ell,f}   \lambda_{\ell,f} ,\notag \\
& \Delta R(\delta , \pbf) \leq \sum\limits_{\ell =1 }^{n} \ell\binom{n}{\ell} \sum\limits_{f =1}^m \rho_{\ell,f}^*  \lambda_{\ell,f}^* \delta \notag \\
& \quad + 
 n \sum\limits_{f =1}^m q_f (1-p_{f}M)\Big(1 - \prod\limits_{f' =1}^{m}  (1-p_{f'}M)^{{(\Gbf-{\bf I})}_{f'f}}\Big)  \delta,
\notag \\
&\bar m \triangleq   \sum_{f =1}^m ( 1- (1-q_f)^{n} ), \notag
\end{align}
with  
\begin{align}
&\lambda_{\ell,f} \triangleq  \Big(\prod\limits_{f'=1}^m (1-p_{f'}M)^{(n-\ell+1)\Gbf_{f'f}} \Big)\notag \\ 
& \qquad\qquad\qquad\qquad\quad\; \Scale[0.99]{  \bigg(1-\prod\limits_{f'=1}^m\Big(1-(p_{f'}M)^{\ell-1}\Big)^{\Gbf_{f'f}} \bigg) } ,\notag
\end{align}
\begin{align}
& \rho_{\ell,f}  \triangleq \mathbb P \Big\{ f = \argmax\limits_{d \in \mathcal D} \lambda_{\ell,d} \Big\} ,\notag\\
& \lambda_{\ell,f}^*  \triangleq     \Big( \prod\limits_{f' =1}^m (1-p_{f'}M)^{(n-\ell+1)\Gbf_{f'f}} \Big)  \notag \\
&  \quad \Scale[0.99]{ \Big(1-(p_{f}M)^{\ell-1}\Big) \bigg(1-\prod\limits_{f'=1}^m\Big(1-(p_{f'}M)^{\ell-1}\Big)^{(\Gbf-{\bf I} )_{f'f} } \bigg)} ,\notag \\
&\rho_{\ell,f}^* \triangleq  \mathbb P \Big\{ f = \argmax\limits_{d \in \mathcal D} \lambda^*_{\ell,d} \Big\} ,  \notag 
\end{align}
$\mathcal D$ denoting a random set of $\ell$ elements selected in an i.i.d. manner from $[m]$, and $\bf I$ denoting the identity matrix.
\end{theorem}

\begin{proof}
For proof of Theorem \ref{thm:general CA-RAP/CM}  see \cite{longerVersion}.

\end{proof}



The CA-RAP caching distribution is computed as the minimizer of the corresponding rate upper bound, $\pbf^*=\argmin_{\pbf} R({\delta},\pbf)$, resulting in the optimal CA-RAP/CM rate $R({\delta},\pbf^*)$. 
 The resulting distribution $\pbf^{*}$ may not have an analytically tractable expression in general, but  can be numerically optimized for the specific library realization.
We remark that the rate upper bound is derived for a given correlation parameter $\delta$,  whose  value can also  be optimized to minimize the achievable expected rate $R({\delta},\pbf)$.



\section{Simulations and Discussions}~\label{sec:Simulations}
We numerically compare the performance of the polynomial-time CA-RAP/CM scheme given in Sec. \ref{subsec:correlation aware rate}, with existing correlation-unaware schemes: Local Caching with Unicast Delivery (LC/U), Local Caching with Naive Multicast Delivery (LC/NM), and RAP Caching with Coded Multicast Delivery (RAP/CM) \cite{ji15order}. 
LC/U consists of the conventional LFU caching policy, in which the $M$ most popular files are cached at each receiver, followed by unicast delivery. 
LC/NM combines LFU caching with naive multicast delivery, where a common uncoded stream of data packets is simultaneously received and decoded by multiple receivers. RAP/CM is the combination of RAP caching and coded multicasting as presented in \cite{ji15order}, and is equivalent to CA-RAP/CM when $\delta=0$.

We consider a broadcast caching network with $n=10$ receivers  and $m=100$ files requested according to a Zipf distribution with parameter $\alpha$ given by $q_{f} = {f^{-\alpha}}/{\sum_{j=1}^{m}j^{-\alpha}}, \; \forall f \in [m]$.
Under demand distribution $\qbf$, we assume a match matrix $\Gbf$ such that, for each  $(f, f') \in [m]^2$,  $\Gbf_{f,f'}= B_{\rm match}$.
Fig. \ref{fig:alpha 0.8} displays the rates for $\alpha=0.8$, $\delta = 0.2$, and $mB_{\rm match}=4$, as $M$ varies from $0$ to $100$.
We observe that the correlation-aware scheme is able to achieve rate reductions that go well beyond the state of the art correlation-unaware counterpart. 
Specifically, CA-RAP/CM  achieves a $2.7\times$ reduction in the expected rate compared to LC/U and a $2.4\times$ reduction compared to RAP/CM for $M=10$.

\begin{figure}
  \centering
  \includegraphics[width=2.8in]{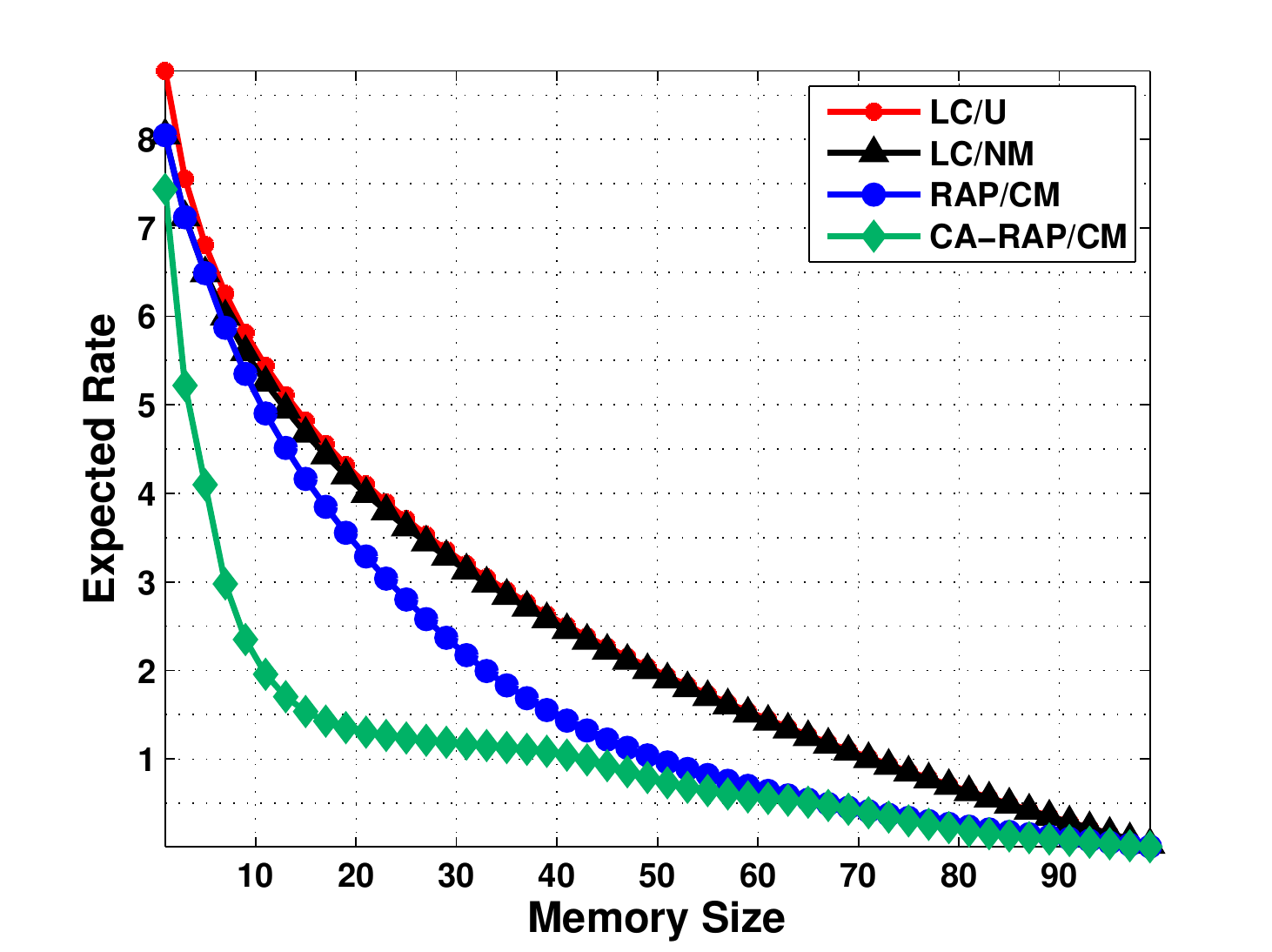}
  \caption{Performance of CA-RAP/CM in a network with Zipf parameter $\alpha = 0.8$, $n = 10$, $m = 100$, $\delta = 0.2$, and $mB_{\rm match}= 4$.}
  \label{fig:alpha 0.8}
\end{figure}

\section{Conclusion}~\label{sec:Conclusion}
In this work, we have shown how exploiting the correlation among the library files can result in more efficient content delivery over cache-aided networks. We proposed a correlation-aware caching scheme in which receivers store content pieces based on their popularity as well as on their correlation with the rest of the file library in the caching phase, and receive compressed versions of the requested files according to the information distributed across the network and their joint statistics during the delivery phase. 
The proposed scheme is shown to significantly outperform state of the art approaches that treat library files as mutually independent.

Ongoing and future work entail investigating the effect of relevant system parameters that characterize the correlated library, such as  $\delta$ and $\Gbf$, on the achievable expected rate, as well as providing order-optimality results. 
An alternative correlation-aware scheme is proposed in \cite{Hassanzadeh2016ITW} where  
the more relevant content from the jointly compressed library is stored during the caching phase, and the transmissions during the delivery phase ensure perfect reconstruction of the requested content. 


\bibliographystyle{IEEEtran}
\bibliography{ISTC_arxiv_v1}

\end{document}